\documentclass[12pt]{article}

\usepackage[a4paper,includeheadfoot,margin=2.54cm]{geometry}
\usepackage[latin2]{inputenc}
\usepackage[english]{babel}
\usepackage{amsthm}
\usepackage{amsmath}
\usepackage{amssymb}
\usepackage[hyphens]{url}

\theoremstyle{plain}
\newtheorem{theorem}{Theorem}
\newtheorem{lemma}{Lemma}
\newtheorem{proposition}[lemma]{Proposition}
\newtheorem{corollary}[lemma]{Corollary}

\newtheorem{example}{Example}[section]
\theoremstyle{definition}

\newtheorem{remark}{Remark}

\newcommand{\F}{\mathbb{F}}

\newcommand{\Z}{\mathbb{Z}}

\newcommand{\T}{{\rm Tr}}

\begin{document}
\begin{center}
\title[A note on the hull and linear complementary pair of cyclic codes \\
\author [Zohreh Aliabadi,  Tekgül Kalayc\i\\

Sabanc\i~ University, MDBF, Orhanl\i, 34956 Tuzla, Istanbul, Turkey\\
E-mail: \{zaliabadi, tekgulkalayci\}@sabanciuniv.edu\\
\end{center}
\begin{abstract}
The Euclidean hull of a linear code $C$ is defined as $C\cap C^{\perp}$, where $C^\perp$ denotes the dual of $C$ under the Euclidean inner product. A linear code with zero hull dimension is called a linear complementary dual (LCD) code.  A pair $(C,D)$ of linear codes of length $n$ over $\F_q$ is called a linear complementary pair (LCP)  of codes if $C\oplus D=\F_q^n$.  In this paper, we give a characterization of  LCD and LCP of cyclic codes of length $q^m-1$, $m \geq 1$, over the finite field $\F_q$ in terms of their basic dual zeros and their trace representations. We also formulate the hull dimension of a cyclic code of  arbitrary length  over $\F_q$ with respect to its basic dual zero. Moreover, we provide a  general formula for the dimension of the intersection of two cyclic codes of arbitrary length over $\F_q$ based on their basic dual zeros.
\end{abstract}
\section{Introduction}
\label{sec1}
The Euclidean hull of a linear code $C$ over the finite field $\F_q$ is defined as the intersection of $C$ with its dual, i.e.,
$$\hbox{Hull}(C):= C\cap C^{\perp},$$
where $C^{\perp}$ is the Euclidean dual of $C$. Obviously, $\hbox{Hull}(C)$ is also a linear code over $\F_q$. We denote the dimension of $\hbox{Hull}(C)$ by $h(C)$.

The concept of the hull has been introduced by Assumus and Key in \cite{AFFINE} in order to classify finite projective planes. The hull of a linear code has  applications in classical linear codes and quantum error-correction codes, see \cite{COMAUT}, \cite{PERBET}, \cite{ONCOMAUT}, \cite{1DIMAG}. It turns out that the algorithms for determining permutation equivalence between codes, and determining the automorphism group of a linear code are more effective when the size of the hull dimension of the code is small.

A zero-dimensional hull linear code is called linear complementary dual (LCD) code, which is introduced by Massey in \cite{LCD}. If $C$ is an LCD code of length $n$ over $\F_q$, then $C\oplus C^{\perp}=\F_q^n$.  More generally, a pair $(C,D)$ of linear codes of length $n$ over $\F_q$ is called a linear complementary pair (LCP) of codes if $C\oplus D=\F_q^n$. Clearly, if $C$ is an LCD code, then the pair $(C,C^{\perp})$ is LCP.

The study of LCD and LCP of codes  has a cryptographic motivation. It has been shown that certain cryptosystems, which are defined via linear codes, are more secure against side channel attacks (SCA) and fault-injection attacks (FIA) when LCD or LCP of codes are used in their constructions, see \cite{LCDSTRE}, \cite{SCA}, \cite{SCAFIA}.

Due to the above-mentioned applications, codes with small hull dimension (especially one-dimensional hull codes) are studied in the recent literature, see \cite{SEMI}, \cite{onedimhull}, \cite{1DIMAG} and references therein.

Cyclic codes, their hull dimensions, LCD and LCP classes of cyclic codes are studied in the literature. The characterization of LCD and LCP of cyclic codes in terms of their generator polynomials has been given in \cite{LCP} and \cite{LCDCYCLIC}, respectively. The hull of cyclic codes in terms of their generator polynomials has been formulated in \cite{CHULL}. The class of one-dimensional hull cyclic codes has been studied in \cite{onedimhull}, where the authors used the defining set of  a cyclic code to obtain their results. They have also shown that there exist no binary or ternary one-dimensional hull cyclic codes.

This paper presents results on the hull LCD and LCP classes of cyclic codes  with respect to their basic dual zero sets, and organized as follows. In Section 1, we  recall the basic definitions and results on the cyclic codes and polynomials over finite fields. In Sections 3 and 4, LCD and one-dimensional hull cyclic codes of length $q^m-1$ over $\F_q$ are studied, respectively. Moreover, the hull of a cyclic code of arbitrary length $n$ over $\F_q$ is formulated. In Section 5, we study  LCP of cyclic codes of length $q^m-1$ over $\F_q$. Furthermore, a general formula for the $\ell$-intersection pair $(C,D)$ of cyclic codes (i.e., $\dim(C\cap D)=\ell$) of arbitrary length over $\F_q$ is provided.
\section{Preliminaries}\label{PRE}
In this section, we recall basic properties of cyclic codes and polynomials over finite fields.
\subsection{Cyclic Codes}
Throughout the paper, $\F_q$ denotes the finite field of $q$ elements, $q$ is a prime power and $n$ is a positive integer such that  $\gcd(n, q)=1$, where $\gcd(n, q)$ denotes the greatest common divisor of $n$ and $q$.   A linear code over $\F_q$ of length $n$ and dimension $k$ is a $k$-dimensional subspace of $\F_q^n$, and a codeword is an element of the linear code.  For a linear code $C$ of length $n$, the (Euclidean) dual of $C$, which is denoted by $C^{\perp}$,  is defined as 
$$C^{\perp}:=\lbrace x\in \F_q^n \; | \; <c,x>=\sum_{i=0}^{n-1} c_ix_i=0 \; \text{ for all } \; c\in C \rbrace.$$
A linear code over $\F_q$ of length $n$ is called cyclic, if any cyclic shift of a codeword is again a codeword, i.e., $(c_0,\ldots ,c_{n-1})\in C$ implies $(c_{n-1},c_0,\ldots ,c_{n-2})\in C$. Clearly, the dual of a cyclic code is also cyclic.

Let $C$ be a linear code of length $n$ over $\F_q$ and $c$ be a codeword.  If $c=(c_0,\ldots ,c_{n-1})$ is identified with the polynomial $c(x)=\sum_{i=0}^{n-1}c_ix^i$, then the code $C$ can be seen as a subset of the ring $\F_{q}[x]/(x^n - 1)$. If $C$ is a cyclic code, then any cyclic shift of a codeword is also a codeword, i.e.,   the set $\{c(x) \ | \ c \in C\}$ is an ideal of $\F_{q}[x]/(x^n - 1)$.  Since $\F_{q}[x]/(x^n - 1)$ is a principal ideal domain,  any ideal of $\F_{q}[x]/(x^n - 1)$ has a unique monic generator. 

We recall that for a polynomial $f(x) \in \F_q[x]$ with $f(0) \neq 0$, the polynomial $f^*(x)=\frac{1}{f(0)}x^{\deg f(x)}f(\frac{1}{x})$ is called the reciprocal polynomial of $f(x)$, where $\deg{f(x)}$ denotes the degree of $f(x)$.  A polynomial $f(x)$ is called self-reciprocal if $f(x)=f^*(x)$. 
 Let $g(x)$ be the generator polynomial of a cyclic code $C$, i.e., $C=<g(x)>$.  Then $C^\perp=<h^*(x)>$, where $h(x)=\frac{x^n-1}{g(x)}$ and $h^{\ast}(x)$  is the reciprocal polynomial of $h(x)$. The polynomial  $h(x)$ is called the parity check polynomial of $C$.

Since the generator polynomial of a cyclic code of length $n$ over $\F_q$ is a factor of $x^n - 1$, we recall the factorization of $x^n - 1 $ into monic irreducible polynomials over $\F_q$. 

\subsection{Factorization of $x^n-1$}
We recall that  $n$ and $q$ are relatively prime. Let $a$ be a positive integer. Then the  $q$-cyclotomic coset $B_a$ of $a$ modulo $n$ is defined as follows:
\begin{align*}
	B_a:= \lbrace a,aq, \ldots , aq^{\delta_a-1} \rbrace,
	\end{align*}
where $\delta_a$ is the smallest positive integer such that
$aq^{\delta_a}\equiv a \quad (\hbox{mod } n).$ Note that the number $\delta_a$ is the cardinality of $B_a$, and it is denoted by $|B_a|$. Clearly, for two positive integers $a_1,  a_2 $,  either $B_{a_1}=B_{a_2}$ or $B_{a_1}\cap B_{a_2} =\emptyset$. 
 Let $B(n,q)$ be the set of all the $q$-cyclotomic coset leaders modulo $n$. Then ${\cup}_{a \in B(n,q)} B_a =\Z_n.$ That is, the  set of $q$-cyclotomic cosets modulo $n$ forms a  partition of  $\Z_n$.\\
Let $\alpha$ be a primitive $n$-th root of unity over $\F_q$. Then the minimal polynomial $m_{\alpha^i}(x)$ of $\alpha^i$ over $\F_q$ is 
$$m_{\alpha^i}(x)=\prod_{s \in B_i} (x-\alpha^s).$$
By using the above notation, the factorization of $x^n - 1$ into  monic  irreducible factors over $\F_q$ can be given as below: 
\begin{equation}\label{factor}
	x^n-1=\prod_{i\in B(n,q)} m_{\alpha^i}(x).
\end{equation}
\begin{lemma}\label{SELF}
	Let $\alpha$ be a primitive $n$-th root of unity over $\F_q$. Then $m_{\alpha^i}(x) $ is self-reciprocal if and only if $B_i=B_{-i}$.
\end{lemma}
\begin{proof}
	Suppose that $m_{\alpha^i}(x) \in \F_q[x]$ is  self-reciprocal. That is,  $m_{\alpha^i}(x)=m_{\alpha^i}^{\ast}(x)$. This implies that  the set of roots of $m_{\alpha^i}^{\ast}(x)$ is equal to the set of roots of $m_{\alpha^i}(x)$, i.e., $B_i=B_{-i}$.
	\newline
Conversely, assume that $B_i=B_{-i}$. This means that $\alpha^{-i}$ is a root of $m_{\alpha^i}(x)$, hence the minimal polynomial $m_{\alpha^i}^{\ast}(x)$ of $\alpha^{-i}$  divides $m_{\alpha^i}(x)$. Since 
$m_{\alpha^i}^{\ast}(x)$ and $m_{\alpha^i}(x)$ are both monic and irreducible, we obtain that $m_{\alpha^i}^{\ast}(x) =m_{\alpha^i}(x)$, 	 Hence, the polynomial $m_{\alpha^i}(x)$ is self-reciprocal.
\end{proof}
 	Let $\lbrace i_1, \ldots, i_t \rbrace$ be the set of all  $q$-cyclotomic coset leaders modulo $n$, and $T=\lbrace \alpha^{i_j} \; | \; 1\leq j \leq t\rbrace$.  Suppose  $T_1, T_2\subseteq T$  such that $T_1=\lbrace \alpha^{i_j} \; | \; B_j=B_{-j} \rbrace$ and $T_2=T \setminus T_1$. By Lemma \ref{SELF}, for any $\alpha^{i_j} \in T_1$, $m_{\alpha^{i_j}}(x)$ is self-reciprocal. Thus, Equation \eqref{factor} can be rewritten as follows:
 	\begin{align}\label{rewritten}
x^n-1=\prod_{\alpha^{i_j} \in T_1} m_{\alpha^{i_j}(x)} \prod_{\alpha^{i_j}\in T_2} m_{\alpha^{i_j}}(x)m_{\alpha^{i_j}}^{\ast}(x).
\end{align}
\subsection{Trace representation of cyclic codes}
In this subsection,  we consider the case $n=q^m-1$ for some positive integer $m$. Let $C$ be a cyclic code of length $n$ over $\mathbb{F}_q$ with the generator polynomial $g(x)$. Let $\alpha$ be a primitive $n$-th root of unity over $\F_q$ and  $\lbrace {i_1}, \ldots , {i_t} \rbrace$ be  the set of all  $q$-cyclotomic coset leaders modulo $n$. 
Suppose that $h(x)$ is the parity check polynomial of $C$ and  $S\subseteq \lbrace 1, \ldots, t \rbrace$ such that $h^{\ast}(x)= \prod_{j \in S} m_{\alpha^{i_j}}(x)$.
Then the basic dual zero of $C$ is defined as
$$\hbox{BZ}(C^{\perp})=\lbrace \alpha^{i_j} \; | \; j\in S \rbrace.$$
The following theorem gives a trace representation of a cyclic code $C$ of length $q^m - 1$, where $\T_{q^mq^k}$ denotes the relative trace map from $\F_{q^m}$ to $\F_{q^k}$, for a divisor  $k$ of  $ m$.  
\begin{proposition}\cite[Proposition 2.1]{tracerep}\label{tracerepr}
	Let $\alpha$ be a primitive $n$-th root of unity with $n=q^m - 1$.
	 Suppose that $C$ is a cyclic code, where the generator polynomial of $C^\perp$  is equal to $\prod_{j \in S}m_{\alpha^{i_j} }(x)$, i.e.,  $\hbox{BZ}(C^{\perp})=\lbrace \alpha^{i_j} \; | \; j\in S \rbrace$. Then
	$$C=\Bigl\{\Bigl(\sum_{j\in S} \T_{q^mq}(\lambda_j x^{i_j} )\Bigr)_{x\in \F_{q^m}^{\ast}} \; \big | \; \lambda_j \in \F_{q^m} \Bigr\}.$$
\end{proposition}
In connection with the trace representation above, we will use the following theorem for our results. 
\begin{theorem}\label{ZEROTRACE}(\cite{0TRACE}, Theorem 2.5)
	For $1\leq j \leq t$, let $i_j \geq 1$ be positive integers which are in different $q$-cyclotomic cosets modulo $n$, where $n=q^m - 1$. For $\lambda_1, \ldots , \lambda_t \in \F_{q^m}$, 
	\begin{align*}\T_{q^mq}(\lambda_1 x^{i_1}+\cdots+ \lambda_t x^{i_t})=0 \; \; \textit{for all} \; \; x\in \F_{q^m}
		\end{align*}
	if and only if  $|B_j|=\delta_j<m$ and $\T_{q^mq^{\delta_j}}(\lambda_j) = 0$ for all $j = 1, \ldots,  t$.
\end{theorem}
\section{Linear complementary dual cyclic codes}\label{LCD1}
We recall that the hull of a linear code $C$ is defined as $\hbox{Hull}(C)=C\cap C^{\perp}$ and we  denote the dimension of  $\hbox{Hull}(C)$ by $h(C)$. A linear code $C$ is called linear complementary dual (LCD) if $h(C)=0$. The  characterization of an LCD cyclic code of length $q^m-1$ with respect to its basic dual zero is given in the following theorem.

\begin{theorem}\label{LCD}
	Let $C$ be a  cyclic code of length $q^m - 1$ over $\F_q$, $\alpha$ be  a primitive $n$-th root of unity over $\F_q$. Let $\lbrace {i_1},\ldots , {i_t} \rbrace$ be the set of all $q$-cyclotomic coset leaders modulo $n$. Then $C$ is LCD if and only if  $\alpha^{i_j} \in \hbox{BZ}(C^\perp)$ implies that either $B_j=B_{-j}$ or $\alpha^{-i_j} \in \hbox{BZ}(C^\perp)$   for all $1 \leq j \leq t$. 
\end{theorem}
\begin{proof}
Suppose that $C=<g(x)>$ and $C^\perp=<h^*(x)>$.  That is, $g(x)h(x)=x^n - 1$. We also have  $\gcd(g(x),h(x))=1$, since the polynomial $x^n - 1$ has no repeated factors as   $\gcd(n, q)=1$.
	These together imply that 
	$g^{\ast}(x)h^{\ast}(x)=x^n-1$ and $  \gcd(g^{\ast}(x),h^{\ast}(x))=1.$
	Since  $(C^{\perp})^{\perp}=C$ the basic dual zero of $C^{\perp}$ is a set of representatives of the roots of $g(x)$,  which is equal to $\hbox{BZ}(C)$. 
	

Suppose on the contrary that $C$ is LCD and there exists $1\leq j \leq t$ such that $\alpha^{i_j}\in \hbox{BZ}(C^{\perp})$, $B_j\cap B_{-j}=\emptyset$ and $\alpha^{-i_j}\notin \hbox{BZ}(C^{\perp})$. We  without loss of generality  assume that $j=1$. The assumptions  $\alpha^{i_1}\in \hbox{BZ}(C^{\perp})$ and $\alpha^{-i_1}\notin \hbox{BZ}(C^{\perp})$ imply that $m_{\alpha^{i_1}}(x)\mid h^{\ast}(x)$  and $m_{\alpha^{-i_1}}(x)\nmid h^{\ast}(x)$, respectively.  Since $g^*(x)h^*(x)=x^n - 1$ and $\gcd(g^*(x), h^*(x))=1$, we obtain  $m_{\alpha^{-i_1}}(x)\mid g^{\ast}(x)$, consequently $m_{\alpha^{i_1}}(x)\mid g(x)$.	Therefore, $\alpha^{i_1}\in \hbox{BZ}(C^{\perp})\cap \hbox{BZ}(C)$. 
Assume that $\hbox{BZ}(C^{\perp})=\lbrace \alpha^{i_1} \rbrace \cup T_1$ and $\hbox{BZ}(C)=\hbox{BZ}((C^{\perp})^{\perp})=\lbrace \alpha^{i_1} \rbrace \cup T_2$. By Proposition \ref{tracerepr}, the trace representations of $C$ and $C^{\perp}$ are as follows:
	\begin{align*}
	C=\Big\lbrace\Big(\T_{q^mq}(\lambda_1 x^{i_1}+\sum_{\alpha^{i_j}\in T_1} \lambda_j x^{i_j})\Big)_{x\in \F_{q^m}^{\ast}} \; \big | \; \lambda_j \in \F_{q^m} \Big\rbrace,\\
	C^{\perp}=\Big\lbrace \Big(\T_{q^mq}(\beta_1 x^{i_1}+\sum_{\alpha^{i_j}\in T_2} \beta_j x^{i_j})\Big)_{x\in \F_{q^m}^{\ast}} \; \big | \; \beta_j \in \F_{q^m} \Big\rbrace.
	\end{align*}
	We can take $\lambda_h=\beta_l=0$ for all $\alpha^{i_h}\in T_1$,  $\alpha^{i_l}\in T_2$, and  $\lambda_1=\beta_1=\lambda$ such that $\T_{q^mq^{\delta_j}}(\lambda) \neq 0$.  Then we obtain 
	$ c=\big(\T_{q^mq}(\lambda x^{i_1})\big)_{x \in \F_{q^m}^{\ast}} \in C\cap C^{\perp}. $ Since $\T_{q^mq}(\lambda) \neq 0 $, $c \neq 0$ by Theroem \ref{ZEROTRACE}. This contradicts the assumption that $C$ is LCD.

Conversely, assume that $\alpha^{i_j}\in \hbox{BZ}(C^{\perp})$ implies that either $B_j=B_{-j}$ or $\alpha^{-i_j}\in \hbox{BZ}(C^{\perp})$ for all $i_j $, $1\leq j \leq t$.
If $B_j=B_{-j}$, then we have   $m_{\alpha^{i_j}}(x)=m_{\alpha^{i_j}}^{\ast}(x)$ by Lemma \ref{SELF}. 
	If	$B_j\cap B_{-j}=\emptyset$ and $\alpha^{-i_j}\in \hbox{BZ}(C^{\perp})$, then we have $m_{\alpha^{i_j}}(x) \ | \ h(x)$ and  $m_{\alpha^{i_j}}^{\ast}(x) \ | \ h^*(x)$. These together imply that  $h^{\ast}(x)$ is self-reciprocal. That is, 
	$g(x)h^{\ast}(x)=x^n-1 $ and $ \gcd(g(x),h^{\ast}(x))=1.$
	Therefore, $\hbox{BZ}(C^{\perp}) \cap \hbox{BZ}(C) =\emptyset$. As $C  \oplus C^{\perp}=\F_{q^n}$, we have  $\hbox{BZ}(C) \cup \hbox{BZ}(C^{\perp})=T$. Thus, we  assume without loss of generality that $\hbox{BZ}(C^{\perp})=\lbrace \alpha^{i_1},\ldots,\alpha^{i_s}\rbrace$ and $\hbox{BZ}(C)=\lbrace \alpha^{i_{s+1}},\ldots,\alpha^{i_t}\rbrace$. Then by Proposition \ref{tracerepr}, the trace representations of $C$ and $C^{\perp}$ are as follows:
	\begin{align*}
	C&=\Big \lbrace \Big(\T_{q^mq}(\lambda_1x^{i_1}+\cdots+\lambda_s x^{i_s})\Big)_{x\in \F_{q^m}^{\ast}} \; \big| \; \lambda_j \in \F_{q^m}, 1 \leq j \leq s \Big\rbrace,\\
	C^{\perp}&=\Big\lbrace \Big(\T_{q^mq}(\lambda_{s+1}x^{i_{s+1}}+\cdots+\lambda_t x^{i_t})\Big)_{x\in \F_{q^m}^{\ast}} \; \big | \; \lambda_j \in \F_{q^m},  s + 1 \leq j \leq t \Big \rbrace.
	\end{align*}
Suppose on the contrary that $\hbox{Hull}(C) \neq \{0\} $.  Then there exists	$0\ne c \in \hbox{Hull}(C)$,  and  $\lambda_1, \ldots, \lambda_t \in \F_{q^m}$ such that 
\begin{align*}
	c=\Big(\T_{q^mq}(\lambda_1x^{i_1}+\cdots+\lambda_s x^{i_s})\Big)_{x\in \F_{q^m}^{\ast}}=\Big(\T_{q^mq}(\lambda_{s+1}x^{i_{s+1}}+\cdots+\lambda_t x^{i_t})\Big)_{x\in \F_{q^m}^{\ast}}
	\end{align*}
	Equivalently, 
	\begin{align}\label{eq:trace}
		\Big(\T_{q^mq}(\lambda_1x^{i_1}+\cdots+\lambda_s x^{i_s}-\lambda_{s+1}x^{i_{s+1}}-\cdots-\lambda_t x^{i_t})\Big)_{x\in \F_{q^m}^{\ast}}=0.
	\end{align}
	By Theorem \ref{ZEROTRACE}, the equality in \eqref{eq:trace} holds if and only if $|B_j|=\delta_j<m$ and $\T_{q^mq^{\delta_j}}(\lambda_j)=0$ for all $1\leq j \leq t$. We know that the set $\hbox{BZ}(C^{\perp}) \cup \hbox{BZ}(C)$ contains all the leaders of $q$-cyclotomic cosets  modulo $n$, in particular, the coset leader that contains 1. Since the cyclotomic coset that contains 1  has cardinality $m$, we have a contradiction. Hence, $C$ is LCD.
\end{proof}
\begin{remark}
By Theorem \ref{ZEROTRACE}, a cyclic code $C$ is LCD if and only if for any divisor $m_{\alpha^{i_j}}(x)$ of $h^{\ast}(x)$, we have $ m_{\alpha^{i_j}}^{\ast}(x)$  is also a divisor.  Therefore, $C$ is LCD if and only if $h^{\ast}(x)$ is self-reciprocal, which implies
$$g^{\ast}(x)=\frac{x^n-1}{h^{\ast}(x)}=\frac{x^n-1}{h(x)}=g(x)$$
is self-reciprocal.  This has been also observed by Massey in \cite{LCDCYCLIC}.
\end{remark} 

\begin{corollary}\label{cor:sr}
	Let $\lbrace {i_1},\ldots , {i_t} \rbrace$ be the set of all q-cyclotomic coset leaders modulo $n$. If $B_j=B_{-j}$ for all $1\leq j \leq t$, then any cyclic code of length $n$ over $\F_q$ is LCD.
\end{corollary}
\begin{proof}
	Since $B_j=B_{-j}$ for all $1\leq j \leq t$, by  Lemma \ref{SELF}, the polynomial $m_{\alpha^{i_j}}(x)$ is self-reciprocal for any $1\leq j \leq t$. This means that any factor $g(x)$ of $x^n-1$ is self-reciprocal. Thus the corresponding cyclic code $C$ is LCD. 
\end{proof}
\begin{example}

		 If  $q=2$ and $m=9$, then the polynomial $x^9 - 1$ has the following factorization into monic irreducible polynomials over $\F_q$: $x^9-1=(x+1)(x^2+x+1)(x^6+x^3+1). $ Since
		all the factors of $x^9 - 1$ are self-reciprocal, any binary cyclic code of length $9$ is LCD by Corollary \ref{cor:sr}.\\
		 If $q=3$ and $m=10$, then  the polynomial $x^{10} - 1$ has the following factorization into monic irreducible polynomials over $\F_q$: 
		$x^{10}-1=(x + 1)(x + 2)(x^4 + x^3 + x^2 + x + 1)(x^4 + 2x^3 + x^2 + 2x + 1),$
		Since all the factors  of $x^{10} - 1$ are self-reciprocal, any ternary cyclic code of length $10$ is LCD by Corollary \ref{cor:sr}.

\end{example}

\section{One-dimensional hull cyclic codes}\label{ONEHULL}

In this section, we present a condition for a cyclic code to have one-dimensional hull in terms of its basic dual zero set. We will use the following theorem. 

\begin{theorem}\label{genint}(\cite{PLESS}, Theorem 4.3.7)
	Let $C_i$ be a cyclic code of length $n$ over $\F_q$ with the generator polynomial $g_i(x)$ for $i = 1, 2$. Then $C_1 \cap C_2$ has generator polynomial $\hbox{lcm}(g_1(x), g_2(x))$, where $\hbox{lcm}(g_1(x), g_2(x))$ denotes the least common multiple of the polynomials $g_1(x)$ and $g_2(x)$. 
\end{theorem}

We recall that  $\beta \in \F_{q^m}$ is called a normal element over $\F_q$ if the set  $\{\beta, \beta^q, \ldots, \beta^{q^{m-1}}\}$ forms a basis of $\F_{q^m}$ over $\F_q$. We need  the following lemma for the main result of this section.  

\begin{lemma}\label{linindep}
	Let $\beta$ be a normal element of $\F_{q^m}$ over $\F_q$.  Suppose that $k, l $ are positive integers  with $B_k \cap B_l=\emptyset$. Then the vectors $\big(\T_{q^mq}(\beta x^k)\big)_{x \in \F_{q^m}^{\ast}}$ and $\big(\T_{q^mq}(\beta x^l)\big)_{x \in \F_{q^m}^{\ast}}$ are linearly independent over $\F_q$. \\

\end{lemma}
\textbf{Proof}: The proof is by contradiction. 
Suppose that the vectors $\big(\T_{q^mq}(\beta x^k)\big)_{x \in \F_{q^m}^{\ast}}$ and $\big(\T_{q^mq}(\beta x^l)\big)_{x \in \F_{q^m}^{\ast}}$ are linearly dependent over $\F_q$. Then there exist  nonzero $c_1, c_2 \in \F_{q}$ such that 
\begin{align*}
	c_1\big(\T_{q^mq}(\beta x^k)\big)_{x \in \F_{q^m}^{\ast}} +c_2 \big(\T_{q^mq}(\beta x^l)\big)_{x \in \F_{q^m}^{\ast}}=\big(\T_{q^mq}(c_1\beta x^k+c_2 \beta x^l)\big)_{x \in \F_{q^m}^*}=0.
\end{align*}
This implies that there exists $0 \neq a_x \in \F_{q^m}$ such that $c_1x^k\beta+c_2 x^l\beta=(c_1x^k+c_2 x^l)\beta=a_x^q-a_x=b_x$, for each $x\in \F_{q^m}$. Since $\beta$ is a normal element of $\F_{q^m}$ over $\F_q$, the element $b_x$  has a unique expression of the form $b_x=\sum_{t=0}^{m - 1}d_{x,t}\beta^{q^t},$
where $d_{x,t} \in \F_{q}$. This implies that $c_1x^k+ c_2x^l=d_{x, 0}  \in \F_q$ for all $x \in \F_q$. Therefore, 
$\T_{q^mq}((c_1 x^k+c_2 x^l)\beta)=(c_1x^k+c_2x^l)\T_{q^mq}(\beta)=0$.  Since $\beta \in \F_{q^m}$ is normal over $\F_{q}$, $\T_{q^mq}(\beta)\neq 0$,  which means $c_1x^k+c_2x^l=0$ for all $x \in \F_{q^m}. $ Assume without loss of generality that $c_1\neq 0$ and let $x=1$.  Then $c_1=-c_2$. If we let $x =\theta$, where $\theta$ is a primitive element of $\F_{q^m}$, then $c_1x^k+c_2x^l=c_1(\theta^k - \theta^l)=0$. As $c_1 \neq 0$, we obtain $\theta^k - \theta^l=0$.   That is,  $\theta^{k-l}=1$, which  implies that $k \equiv l$ mod ($q^m-1$). This  contradicts the  assumption that $B_k \cap B_l = \emptyset$. Hence, the result follows.

\begin{theorem}\label{ouronedimhull}
For $n=q^m - 1$, let $C$ be a cyclic code of length $n$. Let $\lbrace{i_1},\ldots , {i_t} \rbrace$ be the set of all leaders of $q$-cyclotomic cosets modulo $n$, and $T=\lbrace \alpha^{i_j} \; | \; 1\leq j \leq t\rbrace$. Then ${h}(C)=1$ if and only if 
	the following holds.
	\begin{itemize}
		\item[i)]	There exists a unique $1\leq j \leq t$ such that $|B_j|=1$ and $B_j\cap B_{-j}=\emptyset$.
		\item[ii)] $\hbox{BZ}(C^{\perp})= T_1 \cup \lbrace \alpha^{i_j} \rbrace$, where $T_1 \subset T$ satisfies for any $\alpha^{i_h}\in T_1$, we have  either $B_h=B_{-h}$ or $\alpha^{-i_h}\in T_1$.
		\end{itemize}
\end{theorem}
\textbf{Proof}: Let $n=q^m - 1$, a code $C=<g(x)>$ be  cyclic of length $n$, and $h(x)$ be the parity check polynomial of $C$.  Let $\lbrace{i_1},\ldots , {i_t} \rbrace$ be the set of all leaders of $q$-cyclotomic cosets modulo $n$, and $T=\lbrace \alpha^{i_j} \; | \; 1\leq j \leq t\rbrace$. 

Suppose that i) and ii) hold. We first show that $\hbox{BZ}((\hbox{Hull}(C))^{\perp})=\hbox{BZ}((C\cap C^{\perp})^\perp)=\lbrace \alpha^{i_j} \rbrace.$  By Theorem \ref{genint}, we know that $\hbox{Hull}(C)=C\cap C^{\perp}=<\hbox{lcm}(g(x),h^{\ast}(x))>=<\hbox{lcm}(\frac{x^n-1}{h(x)},h^{\ast}(x))>$. By ii), we can write $h^*(x)= m_{\alpha^{i_j}}(x)t(x),$ where $t(x)=\prod_{\alpha^{i_h}\in T_1} m_{\alpha^{i_h}}(x).$
Similar to the proof of Theorem \ref{LCD}, we can see that $t(x)$ is self-reciprocal. This implies that $h(x)=m^{\ast}_{\alpha^{i_j}}(x)t(x)$, and hence
$\gcd( g(x), h^{\ast}(x))=m_{\alpha^{i_j}}(x)$. Then we have $\hbox{lcm}(g(x), h^*(x))= \frac{g(x)h^*(x)}{\gcd(g(x), h^*(x))}=\frac{x^n - 1}{{m_{\alpha^{i_j}}^*(x)}}$. That is, $\hbox{Hull}(C)=<\frac{x^n-1}{m_{\alpha^{i_j}}^* (x)}>$, which implies that $\hbox{BZ}((\hbox{Hull}(C))^{\perp})=\lbrace \alpha^{i_j} \rbrace.$ By Proposition \ref{tracerepr},  the trace representation of $\hbox{Hull}(C)$ is as follows: 
\begin{align}\label{tracerephull}
\hbox{Hull}(C)=\big\lbrace \big(\T_{q^mq}(\lambda x^{i_j})\big)_{x \in \F_{q^m}^{\ast}} \ | \ \lambda \in \F_{q^m}\rbrace.
	\end{align}
 Let $\beta$ be a normal element of $\F_{q^m}$ over $\F_q$. Then for any $\lambda \in \F_{q^m}$ there exist $c_0, \ldots , c_{m-1} \in \F_{q}$ such that $\lambda=c_0 \beta+c_1 \beta^q+ \cdots + c_{m-1}\beta^{q^{m-1}}.$
Hence,  we get 
\begin{align*}
\big(\T_{q^mq}(\lambda x^{i_j})\big)_{x \in \F_{q^m}^{\ast}}=c_0\big(\T_{q^mq}(\beta  x^{i_j})\big)_{x \in \F_{q^m}^{\ast}}+\cdots +c_{m-1}\big(\T_{q^mq}(\beta^{q^{m-1}}x^{i_j})\big)_{x \in \F_{q^m}^{\ast}}.
\end{align*}
This implies that $\hbox{Hull}(C)$ is spanned by the vectors 
$\big(\T_{q^mq}(\beta^{q^r} x^{i_j})\big)_{x \in \F_{q^m}^{\ast}} $ for    $\ 0 \leq r \leq q^{m - 1}$, i.e.,  
$$\hbox{Hull}(C)= \hbox{span} \big \lbrace \big(\T_{q^mq}(\beta^{q^r} x^{i_j})\big)_{x \in \F_{q^m}^{\ast}} \ | \ 0 \leq r \leq q^{m - 1} \big \rbrace$$
by the equality in \eqref{tracerephull}. 
By i), $|B_j|=1$, i.e.,   $j \equiv jq^r$ mod ($q^m-1$) for all $0 \leq r \leq  q^{m - 1}$. Hence, for all $0 \leq r \leq q^{m - 1}$, the equality 
$\big(\T_{q^mq}(\beta^{q^r} x^{i_j})\big)_{x \in \F_{q^m}^{\ast}}=\big(\T_{q^mq}(\beta x^{i_j})\big)_{x \in \F_{q^m}^{\ast}}$ holds. That is, 
$\hbox{Hull}(C)=\hbox{span}\big\lbrace\big(\T_{q^mq}(\beta x^{i_j})\big)_{x \in \mathbb{F}_{q^m}^{\ast}}\big\rbrace$. Since $\beta$ is a normal element of $\F_{q^m}$ over $\F_q$, we have $\T_{q^mq}(\beta)\neq 0$. Then the vector
$(\T_{q^mq}(\beta x^{i_j})\big)_{x \in \mathbb{F}_{q^m}^{\ast}} \neq 0$ by  Theorem \ref{ZEROTRACE},  
and hence  ${h}(C)=1$.

Conversely, suppose on the contrary that ${h}(C)=1$, and there exist representatives  $i_j$ such that $\hbox{BZ}(C^{\perp})=T_1 \cup \lbrace \alpha^{i_1} , \alpha^{i_2} \rbrace$,  $B_{i_j}\cap B_{-i_j}=\emptyset$, $|B_{i_j}|=1$ and $\{\alpha^{-i_j}\}\ne \hbox{BZ}(C^{\perp})$ for $j=1,2$. Similar to the proof of Theorem \ref{LCD}, we obtain $\hbox{BZ}((\hbox{Hull}(C))^{\perp})=\lbrace \alpha^{i_1},\alpha^{i_2} \rbrace$. By Proposition \ref{tracerepr}, we have the following trace representation of $\hbox{Hull}(C)$:
$$\hbox{Hull}(C)=\big\lbrace \big(\T_{q^mq}(\lambda_1 x^{i_1} +\lambda_2 x^{i_2})\big)_{x \in \mathbb{F}_{q^m}^{\ast}} \ | \ \lambda_1, \lambda_2 \in \mathbb{F}_{q^m} \big\rbrace.$$
Then  we have 
\begin{align*}
	\hbox{Hull}(C)&=\hbox{span}\big\{\big(\T_{q^mq}(\beta^{q^r} x^{i_1})\big)_{x \in \mathbb{F}_{q^m}^*}, \big(\T_{q^mq}(\beta^{q^r} x^{i_2})\big)_{x \in \mathbb{F}_{q^m}^*} \ | \  0 \leq r \leq q^{m - 1} \big \}.\\
	&=\hbox{span}\big\{\big(\T_{q^mq}(\beta x^{i_1})\big)_{x \in \mathbb{F}_{q^m}^*}, \big(\T_{q^mq}(\beta x^{i_2})\big)_{x \in \mathbb{F}_{q^m}^*} \big \},
	\end{align*}
where the last equality follows from the assumption that $|B_{i_1}|=|B_{i_2}|=1$. 
We also have  ${h}(C)=1$ by assumption, which implies  that all the vectors in the  spanning set of $\hbox{Hull}(C)$  are linearly dependent. Using Lemma \ref{linindep},  we obtain $B_{i_1}=B_{i_2}$. Hence, the result follows. 
\begin{corollary}\label{cor:bintern}
	There exist no binary and ternary one-dimensional hull cyclic codes of length $q^m-1$.
\end{corollary}
\begin{proof}
	If $C$ is a one-dimensional hull cyclic code over $\F_q$,  then by Theorem \ref{ouronedimhull},  $\hbox{BZ}(C^{\perp})=T_1 \cup \lbrace \alpha^{i_j} \rbrace$, where $B_j\cap B_{-j}=\emptyset$,  $|B_j|=1$ and  $\alpha^{-i_j}\notin \hbox{BZ}(C^{\perp})$. This implies that  $i_j\equiv 2i_j$ mod ($2^m-1$), when $q=2$.  Thus $(2^m-1)\mid i_j$, i.e, $i_j=2^m - 1$, a contradiction to the assumption that $i_j<2^m - 1$.

	Similarly, we have $i_j\equiv 3i_j$ mod ($3^m-1$), when $q \equiv 3$.  Thus $(3^m-1)\mid 2i_j$, i.e., $\alpha^{2i_j}=1$. Since $i_j < q^m - 1$, we conclude that $\alpha^{i_j}=-1$, which implies that $\alpha^{-i_j}=-1$. Hence, we have $i_j \equiv -i_j \mod (3^m - 1)$, which contradicts the assumption that $B_j\cap B_{-j}=\emptyset$. 
	\end{proof}
\begin{remark}
		The characterization of one-dimensional hull cyclic codes in terms of their defining sets is given in \cite{onedimhull}, whereas our characterization   is given in terms of  basic dual zero sets  of  cyclic codes.  In \cite{onedimhull}, the authors  also obtain   the non-existence result given in Corollary \ref{cor:bintern} as a consequence of their characterization.  
	\end{remark}

Let $C$ be a cyclic code of length $n$, where $\gcd(n, q)=1$. We keep the notation  of Theorem \ref{ouronedimhull}.  Suppose that  for any $\alpha^{i_j} \in T_2$,  $B_j\cap B_{-j}=\emptyset$ and $\alpha^{-i_j}\notin T_2$.  Then similar to the proof of  Theorem \ref{ouronedimhull}, we can write 
	$$\hbox{BZ}(((\hbox{Hull}(C))^{\perp})=T_2.$$
	This means that 
	$$h_{\hbox{Hull}(C)}^{\ast}(x)=\prod_{\alpha^{i_j}\in T_2} m_{\alpha^{i_j}}(x),$$
	where $h_{\hbox{Hull(C)}}(x)$ is the parity check polynomial of the code  $\hbox{Hull}(C)$. As a result, we arrive at the following theorem, which generalizes Theorem \ref{ouronedimhull} to the hull of cyclic codes of arbitrary length. 
\begin{theorem}
	Let $C=<g(x)>$ be a cyclic code of length $n$ over $\F_q$. Let $\lbrace{i_1},\ldots , {i_t} \rbrace$ be the set of all leaders of $q$-cyclotomic cosets modulo $n$, and $T=\lbrace \alpha^{i_j} \; | \; 1\leq j \leq t\rbrace$. Suppose that  $\hbox{BZ}(C^{\perp})=T_1 \cup T_2$,  and  the following holds. 
	\begin{itemize}
		\item[i)] For any $\alpha^{i_j}\in T_1$, either $B_j=B_{-j}$ or $\alpha^{-i_j}\in T_1$.
		\item[ii)] For any $\alpha^{i_j} \in T_2$,  $B_j\cap B_{-j}=\emptyset$ and $\alpha^{-i_j}\notin T_2$.
	\end{itemize}
	Then $\hbox{BZ}((\hbox{Hull}(C)^{\perp})=T_2$ and
	$h(C)=\sum_{\alpha^{i_j}\in T_2} |B_j|.$
\end{theorem}

\section{Linear complementary pair of cyclic codes}\label{LCP}
A pair $(C,D)$ of linear codes of length $n$ over the  finite field $\F_q$ is called linear complementary pair (LCP) of codes if $\F_q^n=C \oplus D$.  The LCP of codes can be considered as a generalization of LCD codes. Namely, if $C$ is an LCD code, then the pair $(C,C^{\perp})$ is LCP.

The following lemma is required  to obtain the main result of this section. 
\begin{lemma}\label{BZINT}
	Let $C$ and $D$ be two cyclic codes of length $n$ over $\F_q$. Then
	$$\hbox{BZ}((C\cap D)^{\perp})=\hbox{BZ}(C^{\perp})\cap \hbox{BZ}(D^{\perp}).$$
\end{lemma}
\begin{proof}
	Let $g_C(x)$, $g_D(x)$ and $g(x)$ denote the generator polynomials of $C$, $D$ and $C\cap D$,  and $h_C(x)$, $h_D(x)$ and $h(x)$ denote their parity check polynomials, respectively. 	Take $\alpha^{i_j}\in \hbox{BZ}((C\cap D)^{\perp})$. Suppose that $0 \neq c\in C \cap D$, and $c(x)$ is the polynomial  corresponding to the codeword $c$. Since $c \in C \cap D$, we have 
	$c(x)h(x)=c(x)h_C(x)=c(x)h_D(x)=0.$ This implies that 	
	$c^{\ast}(x)h^{\ast}(x)=c^{\ast}(x)h^{\ast}_C(x)=c^{\ast}(x)h^{\ast}_D(x)=0$. 
	As $\alpha^{i_j} \in \hbox{BZ}((C\cap D)^{\perp})$, there exists $ q(x) \in \F_q[x]$  such that  $h^{\ast}(x)=m_{\alpha^{i_j}}(x) q(x)$.    Then 
	\begin{align*}
			c^{\ast}(x)m_{\alpha^{i_j}}(x) q(x)=c^{\ast}(x)h^{\ast}_C(x)=c^{\ast}(x)h^{\ast}_D(x)=0.
		\end{align*}
	Since $c(x) \neq 0$, we obtain that  $m_{\alpha^{i_j}}(x) \ | \ h^{\ast}_C(x)$,  and $m_{\alpha^{i_j}}(x) \ | \ h^{\ast}_D(x)$. Hence,  
	$\alpha^{i_j}\in \hbox{BZ}(C^{\perp})\cap \hbox{BZ}(D^{\perp})$, i.e., $\alpha^{i_j}\in \hbox{BZ}((C\cap D)^{\perp}) \subseteq  \hbox{BZ}(C^{\perp})\cap \hbox{BZ}(D^{\perp})$.

We prove the reverse inclusion by contradiction.  Suppose that $\alpha^{i_j}\in \hbox{BZ}(C^{\perp})\cap \hbox{BZ}(D^{\perp})$ and $\alpha^{i_j} \notin \hbox{BZ}((C\cap D)^{\perp})$. Then $m_{\alpha^{i_j}}(x) \nmid  h^*(x)$, i.e.,  there exist $q(x), r(x) \in \F_{q}[x]$ such that  $h^{\ast}(x)=m_{\alpha^{i_j}}(x) q(x)+r(x)$  with $0 \neq r(x)$ and $\deg r(x)< \deg m_{\alpha^{i_j}}(x)$. Take $0 \neq c\in C\cap D$, and consider the corresponding polynomial $c(x)$.  Since $c \in C \cap D$, we have  $c(x)h^{\ast}(x)=ch^{\ast}_C(x)=0$. As  $\alpha^{i_j} \in  \hbox{BZ}(C^{\perp})$,  the polynomial $m_{\alpha^{i_j}}(x)$ divides  $h_C^*(x)$. That is, there exists $q_1(x) \in \F_q[x]$ such that  $h_C^{\ast}(x)=m_{\alpha^{i_j}}(x) q_1(x)$. Then
	$$0=c(x)m_{\alpha^{i_j}}(x) q_1(x)=c(x)m_{\alpha^{i_j}}(x) q(x)+c(x)r(x).$$
	Thus,
	$$c(x)r(x)=c(x)m_{\alpha^{i_j}}(x)(q_1(x)-q(x)),$$
	and as $c(x) \neq 0$, $$ r(x)=m_{\alpha^{i_j}}(x)(q(x)-q_1(x)).$$
	This implies that $m_{\alpha^{i_j}}(x)\mid r(x)$, which contradicts $\deg r(x)<m_{\alpha^{i_j}}(x)$. Therefore $ \alpha^{i_j}\in \hbox{BZ}((C\cap D)^{\perp})$, i.e., $   \hbox{BZ}(C^{\perp})\cap \hbox{BZ}(D^{\perp}) \subseteq   \hbox{BZ}((C\cap D)^{\perp})  $. 
	\end{proof}
The following theorem characterizes the LCP of cyclic codes $(C, D)$ of length $q^m-1$ over $\F_q$ in terms of the basic dual zeros of $C$ and $D$.
\begin{theorem}\label{lcpofcycliccodes}
	Let  $n=q^m - 1$, $C$ and $D$ be cyclic codes  of length $n$.  Let $\lbrace{i_1},\ldots , {i_t} \rbrace$ be the set of all leaders of $q$-cyclotomic cosets modulo $n$, and $T=\lbrace \alpha^{i_j} \; | \; 1\leq j \leq t\rbrace$.  Then the pair $(C,D)$ is an LCP of codes if and only if $\hbox{BZ}(C^{\perp})=T\setminus \hbox{BZ}(D^{\perp})$.
\end{theorem}
\begin{proof}
	Assume that the pair  $(C,D)$ is an LCP of codes, and there exists $\alpha^{i_j}\in \hbox{BZ}(C^{\perp})$  and $\alpha^{i_j}\ \notin T\setminus \hbox{BZ}(D^{\perp})$, i.e.,  $\alpha^{i_j}\in \hbox{BZ}(C^{\perp})\cap \hbox{BZ}(D^{\perp})$. Let $\hbox{BZ}(C^{\perp})=\lbrace \alpha^{i_j}\rbrace \cup T_1$ and $\hbox{BZ}(D^{\perp})=\lbrace \alpha^{i_j} \rbrace\cup T_2$. Similar to the proof of  Theorem \ref{LCD}, we can see that there exists $\lambda \in \F_{q^m}$  such that $\T_{q^mq}(\lambda)\ne 0$ and 
	$0\ne c=\big(\T_{q^mq}(\lambda x^{i_j})\big)_{x\in \F_{q^m}^{\ast}} \in C \cap D.$ This contradicts the assumption that the pair $(C,D)$ is LCP.
	
	 Conversely, suppose on the contrary that $\hbox{BZ}(C^{\perp})=T\setminus \hbox{BZ}(D^{\perp})$ and the pair $(C, D)$ is not LCP.  As $\hbox{BZ}(C^{\perp})=T\setminus \hbox{BZ}(D^{\perp})$, we without loss of generality assume that $\hbox{BZ}(C^{\perp})=\lbrace \alpha^{i_1},\ldots,\alpha^{i_s}\rbrace$ and $\hbox{BZ}(D^{\perp})=\lbrace \alpha^{i_{s+1}},\ldots,\alpha^{i_t}\rbrace$. Let $0\ne c\in C\cap D$. Similar to the proof of Theorem \ref{LCD}, there exist $\lambda_1,\ldots, \lambda_t\in \F_{q^m}$ such that 
	$$c=\big(\T_{q^mq}(\lambda_1x^{i_1}+\cdots+\lambda_s x^{i_s})\big)_{x\in \F_{q^m}^{\ast}}=\big(\T_{q^mq}(\lambda_{s+1}x^{i_{s+1}}+\cdots+\lambda_t x^{i_t})\big)_{x\in \F_{q^m}^{\ast}}.$$
	Thus
	$$\T_{q^mq}(\lambda_1x^{i_1}+\cdots+\lambda_s x^{i_s}-\lambda_{s+1}x^{i_{s+1}}-\cdots-\lambda_t x^{i_t})=0 \quad \text{ for all } x\in \F_{q^m}^{\ast}.$$
	By assumption, $\hbox{BZ}(C^{\perp}) \cup \hbox{BZ}(D^{\perp})$ contains all the leaders of $q$-cyclotomic cosets modulo $q^m-1$, in  particular, the coset leader that contains 1, which is of cardinality $m$. Hence, we obtain a  contradiction to  Theorem \ref{ZEROTRACE}. 
\end{proof}
\begin{remark}
Let $(C,D)$ be an LCP of cyclic codes of length $n$. By Theorem \ref{lcpofcycliccodes}, we have $ \hbox{BZ}(C^{\perp})= T \setminus  \hbox{BZ}(D^{\perp})$.  Then we have  
$$x^n-1=\prod_{\alpha^{i_j}\in\rm{BZ}(C^{\perp})}m_{\alpha^{i_j}}(x) \prod_{\alpha^{i_j}\in \rm{BZ}(D^{\perp})}m_{\alpha^{i_j}}(x)=h_C^{\ast}(x) h_D^{\ast}(x),$$
by Equation \eqref{rewritten}.  That is, 
$$x^n-1=g_C(x)g_D(x).$$
This implies that $g_C(x)$ and $g_D(x)$ are relatively prime, which has been also observed in \cite[Remark 2.3]{LCP}.
\end{remark}
A pair $(C,D)$ of linear codes is called linear $\ell$-intersection pair of codes if $\dim(C\cap D)=\ell$. Note that if $(C, D)$ is an LCP of codes of lengh $n$,  then $(C, D)$ is a  0-intersection pair of codes with $n=\dim(C)+\dim(D)$. We then have the following theorem, which generalizes Theorem \ref{lcpofcycliccodes} to any linear $\ell$-intersection pair of cyclic codes with arbitrary length $n$. 
\begin{theorem}
	Let $C$ and $D$ be cyclic codes of length $n$ over $\F_q$. Then 
	$$\ell=\dim(C\cap D)= \sum_{\alpha^{i_j}\in T_1} |B_j|,$$
	where  $T_1\subseteq T$ such that $\hbox{BZ}(C^{\perp})\cap \hbox{BZ}(D^{\perp})= T_1$. 
\end{theorem}
\begin{proof}
	Let $ \hbox{BZ}(C^{\perp})\cap \hbox{BZ}(D^{\perp})=T_1$. By 
Lemma  \ref{BZINT}, we have
$T_1= \hbox{BZ}((C\cap D)^{\perp})$. 
This means that 
$$h^{\ast}(x)=\prod_{\alpha^{i_j}\in T_1} m_{\alpha^{i_j}}(x),$$
where $h(x)$ is the parity check polynomial of $C\cap D$. Therefore, we obtain 
$$\ell=\dim(C\cap D)=\deg h(x)=\sum_{\alpha^{i_j}\in T_1} |B_j|.$$
\end{proof}
\section*{Acknowledgement}
The authors would like to thank Cem G\"{u}neri for pointing out the problem and
helpful discussions, and Nurdag\"{u}l Anbar for her suggestions that improve the quality of the presentation of the paper. T. K. is supported by T\"UB\.ITAK Project under Grant 120F309.


\begin{thebibliography}{99}
\bibitem{AFFINE} Assmus Jr.  E.F.,  Key J.D., \textit{Affine and projective planes}, Discrete Math., 1990; 83:  161-187.
\bibitem{LCDSTRE} Bhasin S.,  Danger J.L., Guilley S., Najm Z., Ngo X.T., \textit{Linear complementary dual code improvement to strengthen encoded circuit against hardware Trojan horses}, IEEE International Symposium on Hardware Oriented Security and Trust, 2015: 82-87.
\bibitem{SCA} Bringer J., Carlet C., Chabanne H., Guilley S., Maghrebi H., \textit{Orthogonal direct sum masking: a smartcard friendly computation paradigm in a code, with builtin protection against side-channel and fault attacks}, WISTP, Lecture Notes in Computer Science,  Springer, Berlin, Heidelberg, 2014;  8501:    40-56.
\bibitem{SEMI} Carlet C.,  Li C.,  Mesnager S., \textit{Linear codes with small hulls in semi-primitive
case}, Des. Codes Cryptogr., 2019;  87 (12): 3063-3075.
\bibitem{LCP} Carlet C., G\"{u}neri C., \"{O}zbudak F., \"{O}zkaya B., Sol\'{e} P., \textit{On linear complementary pairs of codes}, IEEE Trans. Inform. Theory, 2018;  64:  6583-6589.
\bibitem{SCAFIA} Carlet C., Guilley S., \textit{Complementary dual codes for counter-measures to side-channel attacks}, Adv. Math. Commun., 2014;  10 (1): 131-150.
\bibitem{0TRACE} G\"uneri C., \textit{Artin-Schreier curves and weights of two-dimensional cyclic codes}, Finite Fields Appl., 2004; 10:  481-505.
\bibitem{PLESS} Huffman W.C., Pless W.C., \textit{Fundamentals of Error-Correcting Codes}, Cambridge University Press, Cambridge, 2003.
\bibitem{COMAUT} Leon J.S., \textit{Computing automorphism groups of error-correcting codes}, IEEE Trans. Inform. Theory, 1982;  28 (3):  496-511.
\bibitem{onedimhull} Li C., Zeng P., \textit{Constructions of linear codes with one-dimensional hull}, IEEE Trans. Inform. Theory, 2019; 65 (3): 1668-1676.
\bibitem{LCD} Massey J.L., \textit{Linear codes with complementary duals}, Discrete Math., 1992; 106-107:  337-342 (1992).
\bibitem{PERBET} Sendrier N., \textit{Finding the permutation between equivalent codes: the support splitting algorithm}, IEEE Trans. Inform. Theory, 2000; 46 (4): 1193-1203. 
\bibitem{ONCOMAUT} Sendrier N., Skersys G., \textit{On the computation of the automorphism group of a linear code}, Proc. IEEE Int. Symp. Inf. Theory, 2001; 13. 
\bibitem{1DIMAG} Sok L., \textit{On linear codes with one-dimensional Euclidean hull and their applications to EAQECCs}, IEEE Trans. Inform. Theory, 2022; 68 (7) :  4329-4343.
\bibitem{CHULL} Sangwisut E., Jitman S., Ling S., Udomkavanich P., \textit{Hulls of cyclic and negacyclic codes over finite fields}, Finite Fields Appl., 2015;  33:  232-257.
\bibitem{tracerep} Wolfmann J., \textit{New bounds on cyclic codes from algebraic curves}, in: Lecture Notes in Computer Science,  New York: Springer-Verlag, 1989; 388: 47-62.
\bibitem{LCDCYCLIC} Yang X., Massey J.L., \textit{The condition for a cyclic code to have a complementary dual}, Discrete Math., 1994;  126:  391-393.
\end{thebibliography}
\end{document}